\documentclass[runningheads,orivec]{llncs}

\usepackage{balance}
\usepackage{mathrsfs}
\usepackage{multirow}
\usepackage{multicol}
\usepackage[vlined,ruled,linesnumbered,boxed]{algorithm2e}
\usepackage[OT2,OT1]{fontenc}
\usepackage{amssymb}
\usepackage{color}
\usepackage{amstext}
\usepackage{hyperref}
\usepackage{array}
\usepackage{graphicx}
\begin{document}
\renewcommand\arraystretch{1.2}
\title{Characterizing Triviality of the Exponent Lattice of A Polynomial through Galois and Galois-Like Groups\thanks{This work is supported partly by NSFC under grants 61732001 and 61532019. }}
\titlerunning{Lattice Triviality through Groups}
%
\author{Tao Zheng}
%
\authorrunning{T. Zheng}
%
\institute{School of Mathematical Sciences, Peking University\\Beijing, China
\\
\email{xd07121019@126.com}}
\maketitle 
%
\begin{abstract}
The problem of computing \emph{the exponent lattice} which consists of all the multiplicative relations between the roots of a univariate polynomial has drawn much attention in the field of computer algebra. As is known, almost all irreducible polynomials with integer coefficients have only trivial exponent lattices. However, the algorithms in the literature have difficulty in proving such triviality for a generic polynomial. In this paper, the relations between the Galois group (respectively, \emph{the Galois-like groups}) and the triviality of the exponent lattice of a polynomial are investigated. The $\bbbq$\emph{-trivial} pairs, which are at the heart of the relations between the Galois group and the triviality of the exponent lattice of a polynomial, are characterized. An effective algorithm is developed to recognize these pairs. Based on this, a new algorithm is designed to prove the  triviality of the exponent lattice of a generic irreducible polynomial, which considerably improves a state-of-the-art algorithm of the same type when the polynomial degree becomes larger. In addition, the concept of the Galois-like groups of a polynomial is introduced. Some properties of the Galois-like groups are proved and, more importantly, a sufficient and necessary condition is given for a polynomial (which is not necessarily irreducible) to have trivial exponent lattice.

\end{abstract}
\section{Introduction}

Set $\overline{\bbbq}^*$ to be the set of nonzero algebraic numbers. Suppose that $n\in\bbbz_{>0}$. For any $v\in(\overline{\bbbq}^*)^n$, define \emph{the exponent lattice} of $v$ to be
$
\mathcal{R}_v=\{u\in\bbbz^n\;|\;v^u=1\},
$
where $v^u=\prod_{i=1}^nv(i)^{u(i)}$ with $v(i)$ the $i$-th coordinate of $v$ and $u(i)$ the one of $u$. For a univariate polynomial $f\in\bbbq[x]$ (with $f(0)\ne 0$) of degree $n$, denote by $\vec\Omega\in(\overline{\bbbq}^*)^n$ the vector formed by listing all the complex roots of $f$ with multiplicity in some order. For convenience, we call $\mathcal{R}_{\vec\Omega}$ \emph{the exponent lattice} of the polynomial $f$ and use the notation $\mathcal{R}_f$ instead of $\mathcal{R}_{\vec\Omega}$, if no confusion is caused. Moreover, we define $\mathcal{R}_f^\bbbq=\{u\in\bbbz^n\;|\;\vec\Omega^u\in\bbbq\}$.

The exponent lattice has been studied extensively from the perspective of   number theory and algorithmic mathematics since the year 1997 (see \cite{ge1993,kauers2005,loxton1983multiplicative,masser1988linear} and \cite{matveev1994linear,pappalardi2018multiplicatively,van1977multiplicative,zheng2019computing,issac}). There are applications of the exponent lattice to many other areas or problems  concerning, for example, linear recurrence sequences, loop invariants, algebraic groups, compatible rational functions and difference equations (see \cite{almagor2018effective,chen2011structure,derksen2005quantum,kauers2005,lvov2010polynomial,structure}).  Many of the applications involve computing the exponent lattice of a polynomial in $\bbbq[x]$. A lattice $\mathcal{R}\subset\bbbz^n$ or a linear subspace $\mathcal{R}\subset\bbbq^n$ is called \emph{trivial} if any $v\in\mathcal{R}$ satisfies $v(1)=\cdots=v(n)$. It is proved in \cite{drmota1995relations} that almost all irreducible monic polynomials in $\bbbz[x]$ have trivial exponent lattices. However, the state-of-the-art algorithms ({\tt FindRelations} in \cite{ge1993,kauers2005} and {\tt GetBasis} in \cite{issac}), which compute the lattice $\mathcal{R}_v$ for a general $v\in(\overline{\bbbq}^*)^n$, have difficulty in proving the triviality of $\mathcal{R}_f$ for an irreducible monic polynomials $f$ in $\bbbz[x]$. Recently, an algorithm called {\tt FastBasis} is introduced in \cite{zheng2019computing} to efficiently prove the triviality of the exponent lattice of a given \emph{generic} polynomial. 

In Section \ref{2}, the relations between the Galois group and the  triviality of the exponent lattice of an irreducible polynomial is studied. By characterizing the so called $\bbbq$-\emph{trivial} pairs from varies points of view (Proposition \ref{trivial}, \ref{equi} and \ref{gt}), we design an algorithm (Algorithm \ref{IsQtrivial}) recognizing all those $\bbbq$-trivial pairs derived from transitive Galois groups. Base on this, an algorithm called {\tt FastBasis}$_+$ is obtain by  adjusting the algorithm {\tt FastBasis} in \cite{zheng2019computing}. It turns out that {\tt FastBasis}$_+$ is much more efficient than {\tt FastBasis} in proving the triviality of the exponent lattice of a generic irreducible polynomial when the degree of the polynomial is large.

In Section \ref{3}, we define the Galois-like groups of a polynomial since the Galois group of a polynomial does not contain enough information to decide whether the exponent lattice is triviality or not (Example \ref{not-enough}). We prove that a Galois-like group of a polynomial is a subgroup of the automorphism group of the multiplicative group generated by the polynomial roots (Proposition \ref{auto}). Furthermore, almost all conditions on the Galois group assuring the triviality of the exponent lattice can be generalized to correspondent ones on the Galois-like groups (see \S\,\ref{3.2}). More importantly, a sufficient and necessary condition is given for a polynomial (not necessarily irreducible) to have trivial exponent lattice through the Galois-like groups (see Theorem \ref{rftri} and \ref{rfqtri}).
\section{Lattice Triviality Through Galois Groups}\label{2}
\subsection{$\bbbq$-Triviality Implying Lattice Triviality}
Set $G$ to be a finite group, $H$ a subgroup of $G$. Set $\overline g=gH$ for any $g\in G$, then $G$ can be regarded as a permutation group on the set of the left co-sets $G/H=\{\overline{g}\;|\;g\in G\}$ via acting $s\overline{g}=\overline{sg}$. The pair $(G,H)$ is called \emph{faithful, primitive, imprimitive, doubly transitive, doubly homogeneous, etc.}, when the permutation representation of $G$ on $G/H$ has the respective property. The group algebra $\bbbq[G]=\{\sum_{s\in G}a_ss\;|\;s\in G,a_s\in\bbbq\}$ is defined as usual and the $\bbbq$-vector space $\bbbq[G/H]=\big\{\sum_{\overline s\in G/H}a_{\overline s}\overline s\,|\,s \in G,a_{\overline s}\in\bbbq\big\}$ becomes a left $\bbbq[G]$-module via acting $\lambda\overline{t}=\sum_{s\in G}a_s\overline{st}$, with $\lambda=\sum_{s\in G}a_ss\in\bbbq[G]$ and $\overline{t}\in G/H$. We set $\bbbz[G/H]=\big\{\sum_{\overline s\in G/H}a_{\overline s}\overline s\,|\,s \in G,a_{\overline s}\in\bbbz\big\}$ for convenience.

A subset $M$ of $\bbbq[G/H]$ is called $\bbbq$-\emph{admissible} if there is an element $\mu\in\bbbq[G] $ with stabilizer $G_\mu=H$ so that $m\mu=0$ for any $m\in M$ (Definition 3 of \cite{girstmair1999linear}).  Set $\mathcal{V}_1=\big\{a\sum_{\overline{s}\in G/H}\overline{s}\;|\;a\in\bbbq\big\}$. Then the pair $(G,H)$ is called $\bbbq$-\emph{trivial} if $0$ and $\mathcal{V}_1$ are the only two $\bbbq[G]$-submodules that are $\bbbq$-admissible (Definition 7 of \cite{girstmair1999linear}). A polynomial $f\in\bbbq[x]$ ($f(0)\neq0$) without multiple roots is called \emph{non-degenerate} if the quotient of any two roots of $f$ is not a root of unity, and \emph{degenerate} otherwise. The relations between the $\bbbq$-triviality of a pair and the triviality of an exponent lattice is given below:
\begin{proposition}\label{trivial}
Suppose that $L$ is a finite Galois extension of the field $\bbbq$ with Galois group $G$, and that $H<G$ is a subgroup of $G$ so that $G$ operates on the set $G/H$ faithfully. Then the pair $(G,H)$ is $\bbbq$-trivial iff for any $f\in\bbbq[x]$ $(f(0)\neq0)$ satisfying all the following conditions, $\mathcal{R}_f$ is trivial:

 (i) $f$ is irreducible over $\bbbq$ and non-degenerate; 
 
 (ii) the splitting field of $f$ equals $L$ and its Galois group $G_f=G$; 
 
 (iii)  $H$ is the stabilizer of a root of $f$.
\end{proposition}
\begin{proof}
``If'': Suppose on the contrary that the pair $(G,H)$ is not $\bbbq$-trivial. Then there is a $\bbbq$-admissible $\bbbq[G]$-submodule $M$ of $\bbbq[G/H]$ containing an element $v=\sum_{\overline{s}\in G/H}v_{\overline{s}}\overline{s}\in\bbbz[G/H]$ so that there are $\overline{s}_1\neq\overline{s}_2\in G/H$ satisfying $v_{\overline{s}_1}\neq v_{\overline{s}_2}$. By definition, there is an element $\mu\in\bbbq[G]$ with $G_{\mu}=H$ such that $v\mu=0$. Since $\bbbq$ is an algebraic number field, \cite{girstmair1999linear} Proposition 4 indicates that $M$ is admissible in the multiplicative sense. Now by \cite{girstmair1999linear} Proposition 2, there is an algebraic number $\alpha\in L^*$ with stabilizer $G_\alpha=H$ and the element $v$ is a non-trivial multiplicative relation between the conjugations of $\alpha$. What's more, any quotient of two conjugations of $\alpha$ cannot be a root of unity. These mean that the minimal polynomial $f$ of $\alpha$ over the field $\bbbq$ is non-degenerate and the lattice $\mathcal{R}_f$ is nontrivial. Denote by $F$ the splitting field of $f$ over $\bbbq$. Then $F$ is a subfield of $L$ and the Galois group $G_f$ of $f$ is isomorphic to $ G/Gal(L/F)$. Since $G_\alpha=H$, $G_{g(\alpha)}=gHg^{-1}$ for any $g\in G$. Hence the fixed field of the group $gHg^{-1}$ is $\bbbq[g(\alpha)]$. Note that $\bbbq[g(\alpha)]\subset F$, $gHg^{-1}\supset${ Gal}$(L/F)$ by Galois theory. Thus $\cap_{g\in G}\;gHg^{-1}\supset${ Gal}$(L/F)$. Since the subgroup $\cap_{g\in G}$ $gHg^{-1}$ of $G$ operates trivially on the set $G/H$ and the group $G$ operates faithfully on this set, $\cap_{g\in G}$ $gHg^{-1}=1$. Hence {Gal}$(L/F)=1$ and $G_f\simeq G$. In fact, $L=F$ and $G_f=G$.  So the existence of $f$ leads to a contradiction.

``Only If'': Assume that there is an irreducible non-degenerate polynomial $f\in\bbbq[x]$ ($f(0)\neq0$) satisfying the condition (iii) with splitting field equal to $L$ and exponent lattice $\mathcal{R}_f$ nontrivial. Suppose that the set of the roots of $f$ is $\Omega$ and $\alpha\in\Omega\subset L^*$ is with stabilizer $G_\alpha=H$. Thus there is a bijection $\tau:G/H\rightarrow\Omega, \overline{g}\mapsto g(\alpha)$ through which the permutation representations of $G$ on these two sets are isomorphic and we have the $\bbbz$-module isomorphism $\bbbz^\Omega\simeq \bbbz[G/H]$. By \cite{girstmair1999linear} Proposition 2, the lattice $\mathcal{R}_f\subset \bbbz^\Omega\simeq \bbbz[G/H]$ provides an admissible subset $M$ of $\bbbz[G/H]$ in the multiplicative sense. Then by Proposition 3 and Definition 3 in \cite{girstmair1999linear}, one sees that $M$ is a $\bbbq$-admissible subset. Since $\mathcal{R}_f$ is nontrivial, the $\bbbq[G]$-module generated by $M$ in $\bbbq[G/H]$ is neither $0$ nor $\mathcal{V}_1$. Hence the pair $(G,H)$ is not $\bbbq$-trivial, which is a contradiction.\qed
\end{proof}
For any irreducible non-degenerate polynomial $f\in\bbbq[x]$ with Galois group $G$ and a root stabilizer $H<G$, Proposition \ref{trivial} gives the weakest sufficient condition on the pair $(G,H)$ for $\mathcal{R}_f$ to be trivial (\emph{i.e.}, $(G,H)$ being $\bbbq$-trivial). However, the pair $(G,H)$
does not contain all the information needed to decide whether the lattice $\mathcal{R}_f$ is trivial. This is shown in the following example.
\begin{example}\label{not-enough}
Set $g(x)=x^4-4x^3+4x^2+6$, then $g$ is irreducible in $\bbbq[x]$. By the {\tt Unitary-Test} algorithm in \cite{yokoyama1995finding}, one proves that $g$ is non-degenerate. Set $L$ to be the splitting field of $g$ over the rational field and $G=${ Gal}$(L/\bbbq)$ its Galois group. Denote by \[\alpha=(2.35014\cdots)+\sqrt{-1}\cdot(0.90712\cdots)\] one of the roots of $g$, and set $H=G_\alpha$ to be its stabilizer. Computing with Algorithm 7.16 in \cite{kauers2005}, we obtain $\mathcal{R}_g=\bbbz$$\cdot\,(-2,2,2,-2)^T$, which is nontrivial (thus $(G,H)$ is not $\bbbq$-trivial by Proposition \ref{trivial}). 

Set $f(x)=g(x-1)$, then $f$ is irreducible over $\bbbq$ with splitting field $L$ and Galois group $G$. Moreover, the number $\alpha+1$ is a root of $f$ with stabilizer $H$. We note that the polynomials $g$ and $f$ share the same pair $(G,H)$. However, computing with \cite{kauers2005} Algorithm 7.16, we obtain that the lattice $\mathcal{R}_f=\{\mathbf{0}\}$ is trivial. 
\end{example} 

\subsection{Characterization of $\bbbq$-Triviality from the Perspective of Representation Theory And Group Theory}\label{2.2}
\begin{proposition}\label{equi}
Suppose that $G$ is a finite group and $H<G$ is a subgroup so that $G$ operates faithfully on $G/H$. Denote by $1^G_H$ the character of the permutation representation of $G$ on the set $G/H$. Then the pair $(G,H)$ is $\bbbq$-trivial iff the character $1^G_H-1$ is $\bbbq$-irreducible.
\end{proposition}
\begin{proof}
By \cite{girstmair1999linear} Proposition 12, the pair $(G,H)$ is $\bbbq$-trivial iff $(G,H)$ is primitive and the character $1^G_H-1$ is $\bbbq$-irreducible. By \cite{dixon2005permutation} Theorem 3,  if the character $1^G_H-1$ is $\bbbq$-irreducible, then $(G,H)$ is primitive.\qed
\end{proof}
Throughout the paper, \emph{a root of rational} refers to an algebraic number $\alpha$ such that there is a positive rational integer $k$ ensuring $\alpha^k\in\bbbq$.
\begin{remark}\label{replacedbyror}
In the settings of Proposition \ref{trivial}, when $(G,H)$ is $\bbbq$-trivial, $(G,H)$ is primitive. This is equivalent to the condition that $H$ is a maximal subgroup of $G$. A polynomial $f$ satisfying the conditions (ii) and (iii) in Proposition \ref{trivial} has a root $\alpha$ with stabilizer $G_\alpha=H$ and the fixed field $\bbbq[\alpha]$ of the group $H$ is a minimal intermediate field of the extension $L/\bbbq$ by Galois theory. Note that $f$ is irreducible over $\bbbq$. If $f$ is degenerate with no root being a root  of rational, then there is an integer $k\neq 0$ so that $1<\deg(\alpha^k)<\deg(\alpha)$. Thus $\bbbq\subsetneqq
\bbbq[\alpha^k]\subsetneqq\bbbq[\alpha]$, which contradicts the minimality of the field $\bbbq[\alpha]$. Thus, when $(G,H)$ is $\bbbq$-trivial, a polynomial $f$ satisfying the conditions (ii) and (iii) is either non-degenerate or with all roots being roots of rational. Hence the condition (i) in Proposition \ref{trivial} can be replaced by the condition that ``$f$ is irreducible over $\bbbq$ with no root being a root of rational''.
\end{remark}

\begin{proposition}\label{gt}
Let $(G,H)$ be as in Proposition \ref{equi}. Then, regarded as a permutation group operating on the set $G/H$, the group $G$ satisfies exactly one of the following conditions iff the pair $(G,H)$ is $\bbbq$-trivial: 

(i) $G$ is doubly transitive; 

(ii) $G$ is of affine type (but not doubly transitive) of degree $p^d$ for some prime $p$ and $G=M\rtimes H$, where $M\simeq\bbbf_p^d$ is the socle of $G$ and the subgroup $H$ is isomorphic to a subgroup of $GL(d,p)$; moreover, let $Z$ be the center of the group $GL(d,p)$ and regard $H$ as a subgroup of $GL(d,p)$, the group $HZ/Z$ is a transitive subgroup of $PGL(d,p)$ operating on the projective points; 

(iii) $G$ is almost simple (but not doubly transitive) of degree $\frac{1}{2}q(q-1)$, where $q=2^f\geq8$ and $q-1$ is prime, and either $G=PSL_2(q)$ or $G=P\Gamma L_2(q)$ with the size of the nontrivial subdegrees $q+1$ or $(q+1)f$, respectively.
\end{proposition}
\begin{proof}
This is a combination of Theorem 3 and Theorem 12 in \cite{dixon2005permutation} together with Corollary 1.6 in \cite{bamberg2013classification}.\qed
\end{proof}
Denote by $\mathcal{P}$ the set of prime numbers and by $\mathcal{P}^\omega$ the set of prime powers $\{p^d|\,p\in\mathcal{P},d\in\bbbz_{\geq1}\}$. A useful corollary is as follows:
\begin{corollary}\label{corosp}
Suppose that a polynomial $f\in\bbbq[x]$ $(f(0)\neq0)$ is irreducible with Galois group $G$ and a root stabilizer $H$. If the number $\deg(f)$ is NOT in the set 
\begin{equation}\label{sp}
\mathcal{S}=\mathcal{P}^\omega\cup\big\{2^{f-1}(2^f-1)\,\big|\,f\in\bbbz_{\geq3},2^f-1\in\mathcal{P}\big\},
\end{equation}
then the pair $(G,H)$ is $\bbbq$-trivial iff it is doubly transitive.
\end{corollary}
\subsection{Particular $\bbbq$-Trivial Pairs}\label{2.3}

Besides the doubly transitive pairs $(G,H)$, the author provided some other particular $\bbbq$-trivial pairs in \cite{girstmair1999linear} Proposition 13--15. A permutation group $G$ on a set $S$ is called \emph{doubly homogeneous} if for any two subsets $\{s_1,s_2\},\{t_1,t_2\}$ of $S$, there is some  $g\in G$ so that $\{g(s_1),g(s_2)\}=\{t_1,t_2\}$.  In this subsection, we prove that any doubly homogeneous pair $(G,H)$ is also $\bbbq$-trivial.
\begin{proposition}\label{2homoqtri}
Suppose that $G$ is a finite group and $H<G$ is a subgroup so that $G$ operates faithfully on $G/H$. If the  pair $(G,H)$ is doubly homogeneous, then it is $\bbbq$-trivial.
\end{proposition}
\begin{proof}
When the pair $(G,H)$ is doubly transitive, the character $1^G_H-1$ is actually absolutely irreducible. So we are done. Suppose that the pair $(G,H)$ is doubly homogeneous but not doubly transitive, then $G$ is of odd order (Exe. 2.1.11 of \cite{dixon1996permutation}). Then by \cite{feit1963solvability} and {\fontencoding{OT2}\selectfont TEOREMA 7} of \cite{shafarevich1954}, $G$ is the Galois group of a finite Galois extension of the rational field.

Let $f\in\bbbq[x]\;(f(0)\neq0)$ be any polynomial satisfying the conditions (i)--(iii) in Proposition \ref{trivial}. Then $f$ is irreducible and non-degenerate with Galois group $G$. Since the pair $(G,H)$ is doubly homogeneous and the condition (iii) holds, $G$ operates in a doubly homogeneous way on the set $\Omega$ of the roots of $f$. Doubly homogeneousness naturally requires that $\deg(f)=|\Omega|=|G/H|\geq2$. Hence $f$ has no root being a root of rational since it is non-degenerate. By \cite{zheng2019computing} Theorem 3.2, the lattice $\mathcal{R}_f^{\bbbq}$ is trivial and so is the lattice $\mathcal{R}_f$. Finally, according to Proposition \ref{trivial}, the pair $(G,H)$ is $\bbbq$-trivial.\qed
\end{proof}
The following example shows that a $\bbbq$-trivial pair need not be doubly homogeneous.
\begin{example}\label{zhenbaohan}
Set $L$ to be the splitting field of the irreducible polynomial $f=x^5-x^4-4x^3+3x^2+3x-1$ over the rational field, $G$ the Galois group. In fact, $G\simeq C_5$ is the cyclic group of order $5$, and the stabilizer of any root of $f$ is trivial. The faithful pair $(C_5,1)$ is $\bbbq$-trivial by Proposition \ref{p}. Nevertheless, the pair $(C_5,1)$ is not doubly homogeneous.
\end{example}

\begin{proposition}\label{p}
Let $(G,H)$ be as in Proposition \ref{trivial}. If the cardinality of the set $G/H$ is a prime, then $(G,H)$ is $\bbbq$-trivial.
\end{proposition}
\begin{proof}
When $|G/H|=2$, the pair $(G,H)$ is doubly homogeneous and we are done. When $|G/H|$ is an odd prime, the proposition is a straightforward result of \cite{drmota1991multiplicative} Theorem 1 and Proposition \ref{trivial}.\qed
\end{proof}

The figure below shows the relations between different classes of $\bbbq$-trivial pairs. This is based on Theorem 3 of \cite{dixon2005permutation}, Corollary 1.6 of \cite{bamberg2013classification} and Proposition 3.1 of \cite{nottrans}.
\begin{figure}
 \setlength{\abovecaptionskip}{0cm}
\caption{Classification of $\bbbq$-trivial pairs.} \label{fig2}
\centering
 \includegraphics[width=\textwidth]{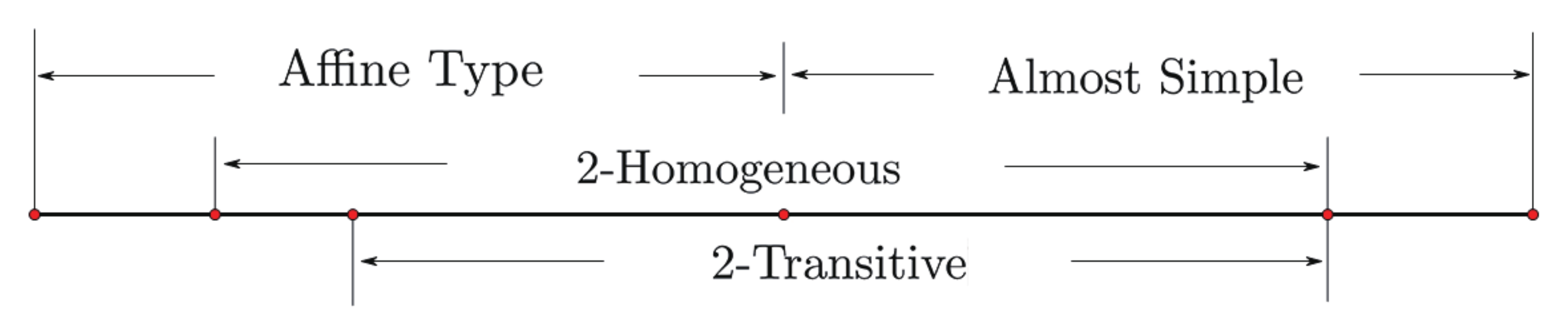}
\end{figure}
\subsection{An Algorithm Deciding $\bbbq$-Triviality of Galois Groups}
Assume that $f\in\bbbq[x]$ $(f(0)\neq0)$ is irreducible with Galois Group $G$ and a root stabilizer $H$. In this subsection, we develop an algorithm deciding whether a pair $(G,H)$ is $\bbbq$-trivial for such a polynomial $f$. Moreover, numerical results show that the algorithm is quite efficient compared with some other relative algorithms (see Table \ref{table3} and \ref{table4}). All numerical results are obtained on a desktop of WINDOWS 7 SYSTEM with 8GB RAM and a 3.30GHz Intel Core i5-4590 processor with 4 cores.
\subsubsection{The ``{{\tt IsQtrivial}}'' Algorithm}
Algorithm \ref{IsQtrivial} shown below is designed according to \S\,\ref{2.2} and \S\,\ref{2.3}. Step 4 of this algorithm is due to Proposition \ref{p} while Step 5 is based on Corollary \ref{corosp}.  The $\bbbq[G]$-submodule $B$ generated by $u$ in Step 6 is contained in the $\bbbq[G]$-submodule \[\mathcal{V}_0=\Big\{\sum_{\overline t\in G/H}a_{\overline t}\overline t\;\Big|\;a_{\overline t}\in\bbbq,\sum_{\overline t\in G/H}a_{\overline t}=0\Big\}\] with character $1^G_H-1$ and $\bbbq$-dimension $\deg(f)-1$. So the correctness of Step 7--10 follows from Proposition \ref{equi}.
\begin{algorithm}
\SetAlgoLined
\KwIn{An irreducible polynomial $f\in\bbbq[x]$ with $f(0)\neq0$;}
\KwOut{``\textbf{True}'' if the pair $(G,H)$ is $\bbbq$-trivial and ``\textbf{False}'' otherwise.}
\caption{{\tt IsQtrivial}}
\textbf{if }$(f$ is reducible or $f(0)==0)$\textbf{ then \{Return }``Error!''\}\textbf{ end}\\
\label{IsQtrivial}
\textbf{if }{$(\deg(f)$ is a prime$)$} \textbf{then \{Return True};\textbf{\} end}\\
Compute the Galois group $G$ of $f$;\\
\textbf{if }{$(G$ is doubly transitive$)$} \textbf{then \{Return True};\textbf{\} end}\\
\textbf{if }{$(\deg(f)\notin\mathcal{S}$ as defined in $(\ref{sp}))$} \textbf{then \{Return False};\textbf{\} end}\\
Compute $B=\,_{\bbbq[G]}$$\langle u\rangle$ with $u=\overline s-\overline1\in\bbbq[G/H]$ for an $s\notin H$;\\
\If {$(\dim(B)==\deg(f)-1$ and $B$ is $\bbbq$-irreducible$)$}
{\textbf{Return True};}
{\textbf{Return False};}
\end{algorithm}

Algorithm \ref{IsQtrivial} is implemented with Magma and random examples are generated to test it. A random polynomial $f$ of degree $n$ with $f(0)\neq0$ is generated in the following way: First, generate its leading coefficient and its constant term by picking integer numbers randomly from the set $\{\pm1,\ldots,\pm10\}$, then pick the rest of the coefficients of $f$ in the set $\{-10,-9,\ldots,10\}$ randomly. Second, check whether $f$ is irreducible: if it is, we are done; otherwise, go back to the first step. The numerical results are shown in Table \ref{table1}.

In Table \ref{table1} (and throughout the section), the notation ``$\#$Poly'' denotes the number of the polynomials that are generated in a single class. As can be seen, almost all the randomly generated polynomials have doubly transitive Galois groups. In fact, most of these Galois groups are symmetry groups. The algorithm is effective and efficient for the randomly generated examples of this kind. In order to test the algorithm for other types of groups, we take advantage of the Magma function {\tt PolynomialWithGaloisGroup}, which provides polynomials with all types of transitive Galois groups of degree between 2 and 15. The results are shown in Table \ref{table2}.

In both tables, the ``GaloisFail'' columns show, for each degree, the numbers of the polynomials with Galois groups computed unsuccessfully in Algorithm \ref{IsQtrivial} Step 3, which is  implemented by the Magma functions {\tt GaloisGroup} and {\tt GaloisProof}. There are more ``GaloisFail'' cases in Table \ref{table2}. The problem is: in those ``GaloisFail'' cases, though the Galois groups can be computed by the first function (which does not provide proven results), the second function returns error and fails to support the result. The ``Average Time'' in Table \ref{table2} excludes the ``GaloisFail'' examples, \emph{i.e.}, it only counts in the ``Qtrivial'' and the ``NotQtrivial'' cases. We see that  the algorithm is still efficient when the Galois group is successfully computed.
\begin{table*}
  \caption{Random Test for {\tt IsQtrivial}}\label{table1}
  \setlength{\abovecaptionskip}{0cm}
  \centering
   \begin{tabular}{|c|c|c|c|c|c|c|}
         \hline
            $\;$Deg$\;$&$\;$\#Poly$\;$&$\;$2-Transitive$\;$&$\;$Qtrivial$\;$&$\;$NotQtrivial$\;$&$\;$GaloisFail$\;$&$\;$Average Time (s)$\;$\\
        \hline
            6&10000&9989&9989&11&0&0.025644\\
      \hline
            8&10000&9998&9998&2&0&0.055090\\
         \hline
         9&10000&10000&10000&0&0&0.069871\\
               \hline
            15&10000&10000&10000&0&0&0.301505\\
                \hline
            20&10000&10000&10000&0&0&0.698264\\     
                            \hline
           28&10000&10000&10000&0&0&1.532056\\
                \hline
            60&40&40&40&0&0&32.2309\\       
      \hline
            81&40&40&40&0&0&107.486\\
      \hline
                  90&40&40&40&0&0&231.638\\
      \hline
                  120&40&40&40&0&0&2057.24\\
      \hline
    \end{tabular}
\end{table*}
\begin{table}
    \setlength{\abovecaptionskip}{0cm}
    \centering
  \caption{Testing {\tt IsQtrivial} by Different Galois Groups}\label{table2}
   \begin{tabular}{|c|c|c|c|c|c|c|}
         \hline
            $\;$Deg$\;$&$\;$\#Poly$\;$&$\;$2-Transitive$\;$&$\;$Qtrivial$\;$&$\;$NotQtrivial$\;$&$\;$GaloisFail$\;$&$\;$Average Time (s)\\
        \hline
            4&5&2&2&3&0&0.019\\ 
        \hline
            6&16&4&4&12&0&0.027\\ 
       \hline
            8&50&6&6&43&1&0.086\\ 
      \hline
            9&34&2&2&23&9&0.095\\ 
      \hline
            10&45&2&2&36&7&0.103\\ 
      \hline
            12&301&2&2&292&7&0.165\\
      \hline
            14&63&2&2&41&20&0.215\\
      \hline
            15&104&2&2&62&40&0.222\\
      \hline
    \end{tabular}
\end{table}
\begin{Large}
\begin{table}
 \setlength{\abovecaptionskip}{0cm}
    \centering 
 \caption{``{\tt IsQtrivial}'' Ensuring Triviality Efficiently}
\begin{tabular}{|c|c|c|c|c|c|c|}
 \hline
 \multirow{3}{*}{$\;$Deg$\;$} &\multirow{3}{*}{$\;$Polynomial$\;$} &\multicolumn{3}{c|}{Runtime (s)}\\
 \cline{3-5}
 &&  \multirow{2}{*}{$\;${\tt FindRelations}$\;$}&\multirow{2}{*}{$\;${\tt GetBasis}$\;$}& \multirow{2}{*}{$\;$\shortstack{{\tt IsQtrivial}\\$+\;${\tt IsROR}}$\;$}\\
  &&&&\\
 \hline
 \multirow{3}{*}{$4$} & $f^{(1)}$  &32.8947&83.7598 & 0.016\\
 \cline{2-5}
&$f^{(2)}$ & 19.1995 &54.343&0.016\\
   \cline{2-5}
  & $f^{(3)}$& 34.6466&90.4592 & 0.016\\
   \hline
 \multirow{3}{*}{$5$} & $g^{(1)}$ & OT&OT & 0.000\\
 \cline{2-5}
& $g^{(2)}$ & OT &OT& 0.000\\
   \cline{2-5}
  &$g^{(3)}$ & OT &OT& 0.000\\
   \hline
 \multirow{3}{*}{$9$} & $h^{(1)}$ & OT&OT & 0.047\\
 \cline{2-5}
& $h^{(2)}$ & OT &OT& 0.078\\
   \cline{2-5}
  & $h^{(3)}$ & OT &OT& 0.047\\
 \hline
 \end{tabular}
 \label{table3}
 \end{table}
 \end{Large}
\subsubsection{Ensuring Lattice Triviality}
By \cite{drmota1995relations} Theorem 2, almost all irreducible polynomials $f$ with $f(0)\ne0$ has trivial lattice $\mathcal{R}_f$. However, the general algorithms, {\tt FindRealtions} in \cite{kauers2005,ge1993} and {\tt GetBasis} in \cite{issac}, dealing with the general input which are arbitrarily given nonzero algebraic numbers instead of all the roots of a certain polynomial, are not every efficient in proving exponent lattice triviality in the latter case.

For a randomly generated irreducible polynomial $f$ with $f(0)\ne0$, if the function {\tt IsQtrivial}$(f)$ returns \textbf{True} and $f$ is proved to have no root being a root of rational by Algorithm 5 in \cite{preprint} (named ``{\tt RootOfRationalTest}'' therein, we call it ``${\tt IsROR}$'' here instead), then $\mathcal{R}_f$ is trivial by Proposition \ref{trivial} and Remark \ref{replacedbyror}. We call this the ``{\tt IsQtrivial}$+${\tt IsROR}'' procedure. Table \ref{table3} shows the efficiency of this procedure to prove the triviality of the exponent lattice of a randomly generated polynomial. The polynomials used here are with integer coefficients picked randomly from the set $\{-10,-9,\ldots,10\}$. 

The notation ``OT'' in Table \ref{table3} (and throughout the section) means the computation is not finished within two hours. As is shown in Table \ref{table3}, it is time consuming for the general algorithms {\tt FingRelations} and {\tt GetBasis} to prove the exponent lattice triviality of a generic polynomial. Thus the ``{\tt IsQtrivial}$+${\tt IsROR}'' procedure can be used before running either of the two general algorithms, when the inputs are all the roots of a certain polynomial. If the procedure fails to prove the triviality, then one turns to the general algorithms.
\subsubsection{The ``{\tt FastBasis}$_+$'' Algorithm}
Similar to the ``{\tt IsQtrivial}$+${\tt IsROR}'' procedure, Theorem 3.2 in \cite{zheng2019computing} allows one to prove lattice triviality of a polynomial by proving doubly homogeneousness of its Galois group and by checking the condition that none of its roots is a root of a rational. Based on this, the algorithm {\tt FastBasis} is designed to compute the lattice $\mathcal{R}_f$ fast for any $f$ in a \emph{generic} set $E\subset\bbbq[x]$ (Definition 5.1 of \cite{zheng2019computing}). 

Similarly, we can define another set $E_{+}\subset\bbbq[x]$ to be the set of polynomials $f$ for which both the following two conditions hold: 

\vspace{0.7mm}
\emph{$(i)$ $\exists c\in\bbbq^*$, $g\in\bbbq[x],k\in\bbbz_{\geq1}$ so that $f=cg^k$, $g$ is irreducible and $x$ does not divide $g(x)$;}

\emph{$(ii)$ all the roots of $g$ are roots of rational} or {\tt {IsQtrivial}}$(g)\;=\;$\textbf{{ True}}. 
\vspace{0.7mm}

\noindent Then, by Proposition \ref{2homoqtri} and Example \ref{zhenbaohan}, one claims that $E_{+}\supsetneqq E$. Thus $E_+$ is also generic in the sense of \cite{zheng2019computing}. Moreover, an algorithm similar to {\tt FastBasis}, which will be called ``{\tt FastBasis}$_+$'', can be obtained by replacing Steps 6--7 in Algorithm 6.1 of \cite{zheng2019computing} (namely, {\tt FastBasis}) by the following step:

\quad...

\textbf{\quad if} ({\tt  IsQtrivial}$(g)==$\textbf{ False }) \textbf{ then }\{\textbf{return} \textcolor{red}{F}\}\textbf{ end if};

\quad...

\noindent Like {\tt FastBasis}, the algorithm {\tt FastBasis}$_+$ computes the lattice $\mathcal{R}_f$ for any $f\in E_+$ while returning a special symbol ``\textcolor{red}{F}'' when $f\notin E_+$.

The algorithm {\tt FastBasis}$_+$ is implemented with Magma while the algorithm {\tt FastBasis} is implemented with Mathematica by the author of \cite{zheng2019computing}. In Table \ref{table4} we compare these two algorithms by applying them to a great deal of random polynomials of varies degree. For an $f\in\bbbq[x]$ of degree at most $n$, we define $h(f)=\max_{0\leq i\leq n}|c_{f,i}|$ with $c_{f,i}$ the coefficient of the term $x^i$ of $f$. The polynomials in Table \ref{table4} are picked randomly from the classes 
\[\bbbz_{10,n}[x]=\{f\in\bbbz[x]\,|\,h(f)\leq 10,\deg(f)\leq n\}.\]
In Table \ref{table4}, the notation ``$\#$Success'' denotes the number of those polynomials in each class for which the algorithm returns a lattice basis successfully within two hours, while the notation ``$\#$F'' gives the number of the polynomials in each class that are proved to be outside the set $E_+$ within two hours.  The average time only counts in all the ``Success'' examples. We can see from Table \ref{table4} and Fig. \ref{fig1} $\,$that for the small inputs with $n<15$, the algorithm {\tt FastBasis} is slightly more efficient while for those lager inputs with $n>15$, the algorithm {\tt FastBasis}$_+$ is much more efficient. This allows one to handle inputs with higher degree that were intractable before.
 \begin{table*}[htbp]
    \setlength{\abovecaptionskip}{0cm}
  \caption{{\tt FastBasis} \emph{v.s.} {\tt FastBasis}$_+$}\label{table4}
    \centering
\begin{tabular}{|c|c|c|c|c|cc|c|c|c|}
\hline
{\multirow{3}{*}{Class}}&\multirow{3}{*}{\#Poly}&\multicolumn{4}{c}{FastBasis}&\multicolumn{4}{|c|}{FastBasis$_+$}\\
\cline{3-10}
&&\multirow{2}{*}{$\#\,$Success}&\multirow{2}{*}{$\,\#\,$F$\,$}&\multirow{2}{*}{$\,$OT$\,$}& \multicolumn{2}{c|}{\multirow{2}{*}{\quad Average Time (s)\quad}}&\multirow{2}{*}{$\,$OT$\,$}&\multirow{2}{*}{$\,\#\,$F$\,$}&\multirow{2}{*}{$\#\,$Success}\\
&&&&&&&&&\\
\hline
{$n=6$} & 10000&9011&989&0&\multicolumn{1}{c|}{$\;\;$0.007304$\;\;$}&$\;\;$0.025499$\;\;$&0&989&9011\\
\hline
{$n=8$} &10000&9064&936&0&\multicolumn{1}{c|}{0.018372}&0.055328&0&936&9064\\
   \hline
{$n=9$} &10000&9113&887&0&\multicolumn{1}{c|}{0.029044}&0.069996&0&887&9113\\
   \hline
{$n=15$} &10000&9227&773&0&\multicolumn{1}{c|}{0.305941}&0.301110&0&773&9227\\
   \hline
{$n=20$}  &10000&9243&757&0&\multicolumn{1}{c|}{1.502110}&0.700131&0&757&9243\\
   \hline
{$n=28$}&10000&9279&721&0&\multicolumn{1}{c|}{9.29961}&1.527806&0&721&9279\\
   \hline
{$n=40$} &100&93&7&0&\multicolumn{1}{c|}{76.3928}&6.788000&0&7&93\\
   \hline
{$n=50$}  &100&96&4&0&\multicolumn{1}{c|}{315.523}&15.70400&0&4&96\\
   \hline
{$n=60$}  &35&33&1&1&\multicolumn{1}{c|}{1291.38}&31.27800&0&1&34\\
   \hline
{$n=81$}  &40&15&1&24&\multicolumn{1}{c|}{5539.67}&104.515&0&1&39\\
   \hline
{$n=90$}  &40&0&2&38&\multicolumn{1}{c|}{--}&224.413&0&2&38\\
   \hline
{$n=120$}  &40&--&--&--&\multicolumn{1}{c|}{--}&2058.228&0&2&38\\
  \hline
 \end{tabular}
 \end{table*}
\begin{scriptsize}
\begin{figure}
 \setlength{\abovecaptionskip}{0cm}
\caption{Comparing average runtime of two algorithms.} \label{fig1}
\centering
 \includegraphics[width=100mm]{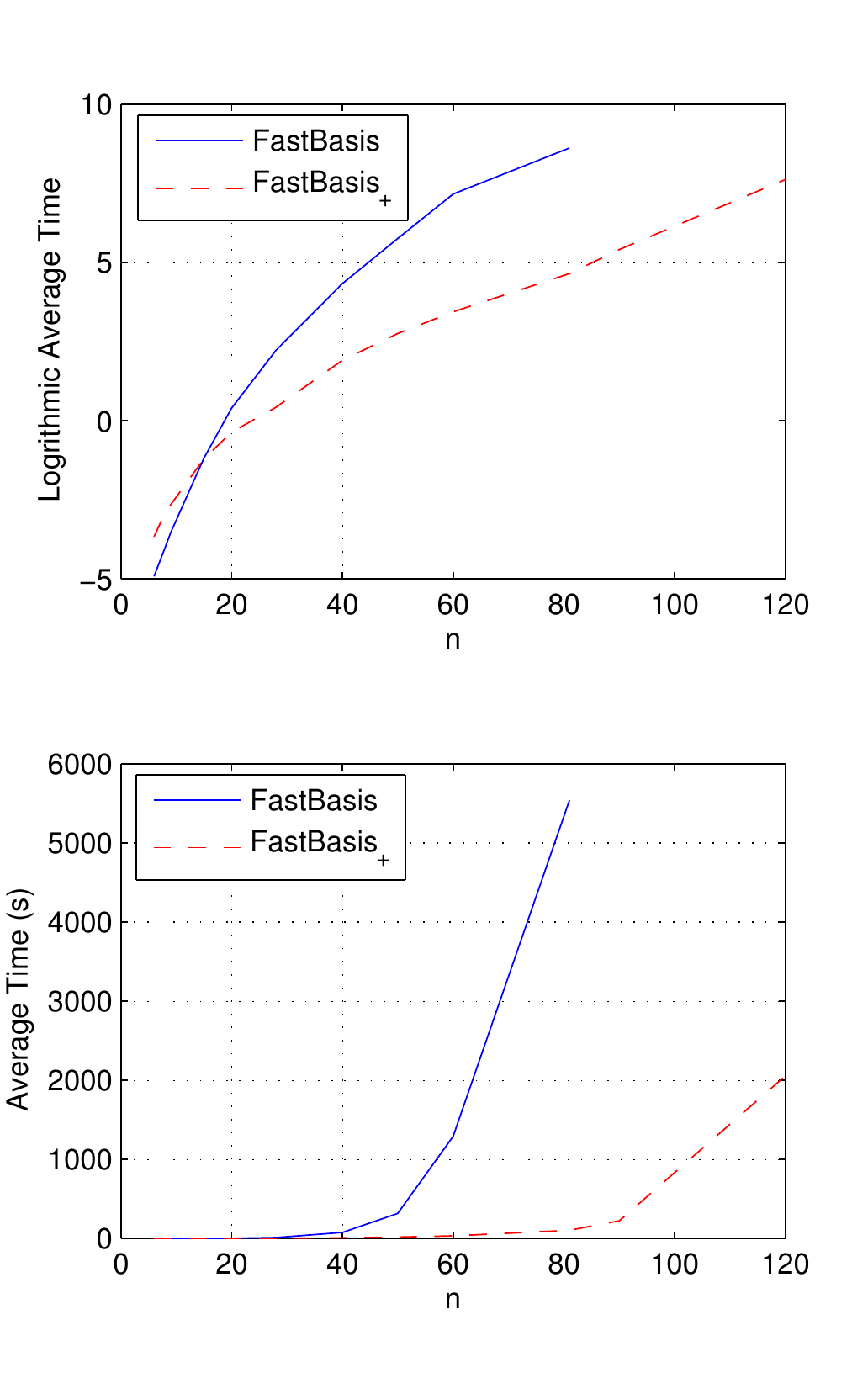}
\end{figure}
\end{scriptsize}
\section{Lattice Triviality through Galois-Like Groups}\label{3}
As is shown in Example \ref{not-enough}, provided only the pair $(G,H)$, one may not be able to decide whether the lattice $\mathcal{R}_f$ is trivial or not. Here $f$ is an irreducible polynomial with Galois group $G$ and a root stabilizer $H$. In this section, the concept of a Galois-like group is introduced. An equivalent condition for the lattice $\mathcal{R}_f$ to be trivial is given through the concept of a Galois-like group. 
\subsection{Root Permutations Preserving Multiplicative Relations}
Set $f\in\bbbq[x]\;(f(0)\neq0)$ to be a polynomial with no multiple roots. Denote by $\Sigma$ the symmetry group operating on the set $\Omega=\{r_1,\ldots,r_n\}$ of the roots of $f$. In the sequel, we denote by $\vec\Omega=(r_1,\ldots,r_n)^T$ a vector of the roots and by $\sigma(\vec\Omega)=(\sigma(r_1),\ldots,\sigma(r_n))^T$ a permutation of $\vec\Omega$ with $\sigma\in\Sigma$. 
\begin{definition}\label{galoislike}
\emph{A Galois-like group} of the polynomial $f$ refers to any one of  the following groups:

(i) $\mathcal{G}_f=\{\sigma\in\Sigma\;|\;\forall v\in\mathbb{Z}^n,\;{\vec\Omega}^v=1\Rightarrow \sigma(\vec\Omega)^v=1\}$;

(ii) $\mathcal{G}_f^B=\{\sigma\in\Sigma\;|\;\forall v\in\mathbb{Z}^n,\;{\vec\Omega}^v\in\bbbq\Rightarrow \sigma(\vec\Omega)^v={\vec\Omega}^v\}$;

(iii) $\mathcal{G}_f^\bbbq=\{\sigma\in\Sigma\;|\;\forall v\in\mathbb{Z}^n,\;{\vec\Omega}^v\in\bbbq\Rightarrow \sigma(\vec\Omega)^v\in\bbbq\}$;
\end{definition}
To verify the terms used above in the definition, we need to prove that any subset of $\Sigma$ defined in Definition \ref{galoislike} is indeed a group:
\begin{proposition}\label{isgroup}
Suppose that $f\in\bbbq[x]$ $(f(0)\neq0)$ is a polynomial with no multiple roots. Set $\mathcal{G}=\mathcal{G}_f,\mathcal{G}_f^B$ or $\mathcal{G}_f^\bbbq$, then $\mathcal{G}$ is a subgroup of $\,\Sigma$.
\end{proposition}
\begin{proof}
Any $\sigma\in\Sigma$ results in a coordinate permutation $\hat{\sigma}$ operating on the space $\mathbb{C}^n$ with $n=|\Omega|$ in a manner so that for any vector $v=(c_1,\ldots,c_n)^T\in\mathbb{C}^n$, $\hat\sigma(v)=(b_1,\ldots,b_n)^T$ with $b_i=c_j$ whenever $\sigma(r_i)=r_j$. Then one observes that the equalities $\widehat{\sigma^{-1}}=(\hat\sigma)^{-1}$, $\sigma(\vec\Omega)^{\hat\sigma(v)}=\vec\Omega^v$ and $\vec\Omega^{\hat\sigma(v)}=\sigma^{-1}(\vec\Omega)^v$ hold for any $\sigma\in\Sigma$ and any $v\in\mathbb{Z}^n$.

Set $\mathcal{G}=\mathcal{G}_f$ ($\mathcal{G}_f^B$ or $\mathcal{G}_f^\bbbq$ respectively) and $\mathcal{R}=\mathcal{R}_f$ ($\mathcal{R}_f^\bbbq$ respectively). Then, by definition, $\hat\sigma^{-1}(v)\in\mathcal{R}$ for any $\sigma\in\mathcal{G}$ and any $v\in\mathcal{R}$. Hence the set $\hat\sigma^{-1}(\mathcal{R})=\{\hat\sigma^{-1}(v)\;|\;v\in\mathcal{R}\}$ is a subset of the lattice $\mathcal{R}$. Noting that $\hat\sigma^{-1}$ operates linearly, one concludes that $\hat\sigma^{-1}(\mathcal{R})$ is also a lattice. Thus $\hat\sigma^{-1}(\mathcal{R})$ is a sub-lattice of $\mathcal{R}$. Since $\hat\sigma^{-1}$ is linear and non-singular, any basis of $\mathcal{R}$ is transformed into a basis of $\hat\sigma^{-1}(\mathcal{R})$ by $\hat\sigma^{-1}$. Hence rank$(\mathcal{R})=$\,rank$(\hat\sigma^{-1}(\mathcal{R}))$. Since $\hat\sigma^{-1}$ is orthogonal on the space $\mathbb{R}^n$ and orthogonal operations preserve the lattice volume, $\mathcal{R}=\hat\sigma^{-1}(\mathcal{R})$. Thus $\hat\sigma(\mathcal{R})=\mathcal{R}$. 

So $\hat\sigma(v)\in\mathcal{R}$ for any $\sigma\in\mathcal{G}$ and any $v\in\mathcal{R}$. If $\mathcal{G}=\mathcal{G}_f$ (or $\mathcal{G}_f^\bbbq$) and $\mathcal{R}=\mathcal{R}_f$ (or $\mathcal{R}_f^\bbbq$ respectively), then $\vec\Omega^{\hat\sigma(v)}=1$ (or $\vec\Omega^{\hat\sigma(v)}\in\bbbq$ respectively). Equivalently, $\sigma^{-1}(\vec\Omega)^v=1$ (or $\sigma^{-1}(\vec\Omega)^v\in\bbbq$). Hence $\sigma^{-1}\in\mathcal{G}$ for any $\sigma\in\mathcal{G}$. Now suppose that $\mathcal{G}=\mathcal{G}_f^B$ and $\mathcal{R}=\mathcal{R}_f^\bbbq$. Since $\hat\sigma(v)\in\mathcal{R}_f^\bbbq$ and $\vec\Omega^{\hat\sigma(v)}\in\bbbq$, $\sigma(\vec\Omega)^{\hat\sigma(v)}=\vec\Omega^{\hat\sigma(v)}$ follows from the definition of $\mathcal{G}_f^B$. The left side of this equality equals $\vec\Omega^v$ while its right side equals $\sigma^{-1}(\vec\Omega)^v$. Hence $\sigma^{-1}(\vec\Omega)^v=\vec\Omega^v$ for any $\sigma\in\mathcal{G}_f^B$ and $v\in\mathcal{R}_f^\bbbq$. Thus $\sigma^{-1}\in\mathcal{G}_f^B$.

The closure of the multiplication in the subset $\mathcal{G}$ of $\Sigma$ and the fact that $1\in\mathcal{G}$ are straightforward. Thus $\mathcal{G}$ is a group. \qed
\end{proof}

Define groups $\langle\Omega\rangle=\{\vec\Omega^v\;|\;v\in\mathbb{Z}^n\}$ and $\langle\Omega\rangle_{\bbbq}=\{c\vec\Omega^v\;|\;c\in\bbbq^*,\;v\in\mathbb{Z}^n\}$. The following proposition asserts that the Galois-like groups $\mathcal{G}_f$ and $\mathcal{G}_f^B$ of a polynomial $f$ are subgroups of the automorphism groups of $\langle\Omega\rangle$ and $\langle\Omega\rangle_\bbbq$ respectively.
\begin{proposition}\label{auto}
The following relations hold:

(i) $\mathcal{G}_f\simeq\{\eta\in$\,Aut$\,(\langle\Omega\rangle)\;|\;\forall r_i\in\Omega,\;\eta(r_i)\in\Omega\}$;

(ii) $\mathcal{G}_f^B\simeq\big\{\eta\in$\,Aut$\,(\langle\Omega\rangle_\bbbq)\;\big|\;\forall r_i\in\Omega,\eta(r_i)\in\Omega;\,\eta|_{\bbbq^*}=id_{\bbbq^*}\big\}$.
\end{proposition}
\begin{proof}
Denote by Aut$_\Omega(\langle\Omega\rangle)$ the group in the right side of the formula in (i) and by Aut$_\Omega^\bbbq(\langle\Omega\rangle_\bbbq)$ the one in the right side of the formula in (ii). 

Set $\mathcal{G}=\mathcal{G}_f$ (or $\mathcal{G}_f^B$) and $A=$\,Aut$_\Omega(\langle\Omega\rangle)$ (or Aut$_\Omega^\bbbq(\langle\Omega\rangle_\bbbq)$ respectively). For any $\sigma\in\mathcal{G}$, we define an element $E_\sigma$ in $A$ in the following way: for any $\vec\Omega^v\in\langle\Omega\rangle$, $E_\sigma(\vec\Omega^v)=\sigma(\vec\Omega)^v$ (or, for any $c\vec\Omega^v\in\langle\Omega\rangle_\bbbq$, $E_\sigma(c\vec\Omega^v)=c\sigma(\vec\Omega)^v$). From the definition of $\mathcal{G}$, we see that $\sigma(\vec\Omega)^v=\sigma(\vec\Omega)^{v'}$ whenever $\vec\Omega^v=\vec\Omega^{v'}\in\langle\Omega\rangle$ (or that $c_1\sigma(\vec\Omega)^v=c_2\sigma(\vec\Omega)^{v'}$ whenever $c_1\vec\Omega^v=c_2\vec\Omega^{v'}\in\langle\Omega\rangle_\bbbq$). Thus the map $E_\sigma$: $\langle\Omega\rangle\rightarrow\langle\Omega\rangle$ (or $E_\sigma$: $\langle\Omega\rangle_\bbbq\rightarrow\langle\Omega\rangle_\bbbq$) is well defined. It is trivial to verify the fact that $E_\sigma$ is an automorphism of $\langle\Omega\rangle$ (or of $\langle\Omega\rangle_\bbbq$) and the property that for all $r_i\in\Omega$, $E_\sigma(r_i)=\sigma(r_i)\in\Omega$ (or, moreover, $E_\sigma(c)=c$ for any $c\in\bbbq^*$). Hence $E_\sigma$ is indeed in the set $A$. Thus $E_{_{^{\,\bullet}}}\,$is a map from $\mathcal{G}$ to $A$.

For any $\eta\in A$, $\eta$ is injective and $\eta(\Omega)\subset\Omega$. Since $\Omega$ is finite, $\eta(\Omega)=\Omega$. Hence $\eta|_\Omega\in\Sigma$. Because $\eta$ is an automorphism (or an automorphism fixing every rational number), $\eta|_\Omega(\vec\Omega)^v=1$ whenever $\vec\Omega^v=1$ (or $\eta|_\Omega(\vec\Omega)^v=\vec\Omega^v$ whenever $\vec\Omega^v\in\bbbq$). Thus $\eta|_\Omega\in\mathcal{G}$. Define $R_\eta=\eta|_\Omega$, then $R_{_{^{\,\bullet}}}\,$is a map from $A$ to $\mathcal{G}$.

It is clear that both the maps $E_{_{^{\,\bullet}}}\,$and $R_{_{^{\,\bullet}}}$are group homomorphisms. That is, $E_{\sigma_1\sigma_2}=E_{\sigma_1}E_{\sigma_2}$ and $R_{\eta_1\eta_2}=R_{\eta_1}R_{\eta_2}$ for any $\sigma_1,\sigma_2\in\mathcal{G}$ and any $\eta_1,\eta_2\in A$. One also verifies easily that $E_{_{R_{_{^{\,\bullet}}}}}$$=id_A$ and $R_{_{E_{_{^{\,\bullet}}}}}$$=id_{\mathcal{G}}$. Hence $\mathcal{G}\simeq A$.\qed
\end{proof}

By definition, a Galois-like group of a polynomial $f$ is the group of the permutations between its roots that preserve all the  multiplicative relations between them. Since any element in the Galois group of $f$ preserves all polynomial relations between the roots, the following relations between the Galois group and a Galois-like group of $f$ is straightforward:
\begin{proposition}\label{subg}
Suppose that $f\in\bbbq[x]$ $(f(0)\neq0)$ has no multiple roots. Then, regarded as a permutation group operating on the roots of $f$, the Galois group of $f$ is a subgroup of any Galois-like group of $f$.
\end{proposition}
Besides, the following relations between the Galois-like groups is straightforward but noteworthy:
\begin{proposition}
Let $f$ be as in Proposition \ref{subg}. Then $\mathcal{G}_f^B\leq\mathcal{G}_f\text{ and }\,\mathcal{G}_f^B\leq\mathcal{G}_f^\bbbq$.
\end{proposition}
\subsection{Generalization of Sufficient Conditions of Exponent Lattice Triviality}\label{3.2}
With the help of the concept of Galois-like groups, we can generalize many sufficient conditions that implying triviality of exponent lattices.
\begin{lemma}\label{gf2tran}
Set $f\in\bbbq[x]$ $(f(0)\neq0)$ to be a polynomial without multiple roots. Denote by $\vec\Omega=(\alpha_1,\ldots,\alpha_s,\gamma_1,\ldots,\gamma_t)^T$ the vector of all the roots of $f$ with $\alpha_i$ the roots that are not roots of rational. Suppose that the Galois-like group $\mathcal{G}_f$ is doubly transitive, then any multiplicative relation $v\in\mathcal{R}_{\vec\Omega}=\mathcal{R}_f$ satisfies the following condition:
\begin{equation}\label{nonrorkave}
v(1)=\cdots=v(s)=\frac{v(1)+\cdots+v(s+t)}{s+t}.
\end{equation}
\end{lemma}
\begin{lemma}\label{gfq2tran}
Let $f$ and $\vec\Omega$ be as in Lemma \ref{gf2tran}. Suppose that the Galois-like group $\mathcal{G}_f^B$ or $\mathcal{G}_f^\bbbq$ is doubly transitive, then any multiplicative relation $v\in\mathcal{R}_{\vec\Omega}^\bbbq=\mathcal{R}_f^\bbbq$ satisfies the condition $(\ref{nonrorkave})$.
\end{lemma}
The proofs of those two propositions above are both almost the same to the one of Theorem 3 in \cite{baron1995polynomial}, because of which we do not give any of them here. A direct corollary of these propositions are as follows:
\begin{proposition}\label{2trantri}
Set $f\in\bbbq[x]$ $(f(0)\neq0)$ to be a polynomial without multiple roots and none of its roots is a root of rational. If the group $\mathcal{G}_f^B$ or $\mathcal{G}_f^\bbbq$ $($respectively, $\mathcal{G}_f)$ is doubly transitive, then the lattice $\mathcal{R}_f^\bbbq$ $($respectively, $\mathcal{R}_f$$)$ is trivial.
\end{proposition}
This is a generalization of Theorem 3 in \cite{baron1995polynomial}. The essential idea is that the proof of Theorem 3 in \cite{baron1995polynomial} relies only on the properties of Galois-like groups (\emph{i.e.}, preserving all the multiplicative relations) but not on those properties that are possessed uniquely by the Galois groups.

Noting that $\mathcal{G}_f^B\leq\mathcal{G}_f$, one concludes form Proposition \ref{2trantri} that both the lattices $\mathcal{R}_f^\bbbq$ and $\mathcal{R}_f$ are trivial whenever the group $\mathcal{G}_f^B$ is doubly transitive and none of the roots of $f$ is a root of rational. More generally, we have the following proposition and Corollary \ref{triiff}:
\begin{proposition}\label{wfwfq}
Set $f\in\bbbq[x]$ $(f(0)\neq0)$ to be a polynomial without multiple roots. Define $\mathcal{W}_f=\bbbq\otimes\mathcal{R}_f$ and $\mathcal{W}_f^\bbbq=\bbbq\otimes\mathcal{R}_f^\bbbq$ with ``$\,\otimes$'' the tensor product of $\,\mathbb{Z}$-modules. Set $\mathcal{V}_0=\{v\in\bbbq^n\;|\;\sum_{i=1}^n v(i)=0\}$ and $\mathcal{V}_1=\{c(1_{_1},\ldots,1_{_n})^T|\;c\in\bbbq\}$ with $n=\deg(f)$. Suppose that $\mathcal{G}_f^B$ is transitive, then the following conclusions hold:

(i) $\mathcal{W}_f^\bbbq\cap\mathcal{V}_0=\mathcal{W}_f\cap\mathcal{V}_0$; 

(ii) $\mathcal{W}_f^\bbbq=\mathcal{W}_f+\mathcal{V}_1$ and thus $\mathcal{W}_f^\bbbq=\mathcal{W}_f$ iff $f(0)\in\{1,-1\}$.
\end{proposition}
\begin{proof}
The proof is almost the same to the one of Lemma 1 in \cite{dixon1997polynomials}, except that we require the transitivity of the Galois-like group $\mathcal{G}_f^B$ instead of the the transitivity of the Galois group of $f$.\qed
\end{proof}
\begin{corollary}\label{triiff}
Let $f$ be as in Proposition \ref{wfwfq} such that the group $\mathcal{G}_f^B$ is transitive, then the lattice $\mathcal{R}_f$ is trivial iff the lattice $\mathcal{R}_f^\bbbq$ is.
\end{corollary}
\begin{proof}
Since the group $\mathcal{G}_f^B$ is transitive, $\mathcal{W}_f^\bbbq=\mathcal{W}_f+\mathcal{V}_1$ by Proposition \ref{wfwfq}. So $\mathcal{W}_f^\bbbq$ is trivial iff $\mathcal{W}_f$ is trivial. Hence 
\[\begin{array}{rcl}
\vspace{3.3mm}
\mathcal{R}_f \text{ is trivial}&\Longleftrightarrow&\mathcal{W}_f \text{ is trivial}\\
&\Longleftrightarrow&\mathcal{W}_f^\bbbq\text{ is trivial}\\
\end{array}\label{above}\]
$\quad\quad\;\quad\quad\quad\quad\quad\quad\quad\quad\quad\quad\quad\quad\quad\quad\Longleftrightarrow\mathcal{R}_f^\bbbq$ is trivial.\qed
\end{proof}
Similar to Proposition \ref{2trantri}, we have the following result:
\begin{proposition}\label{2homotri}
Let $f\in\bbbq[x]$ be as in Proposition \ref{2trantri}. If the group $\mathcal{G}_f^B$ or $\mathcal{G}_f^\bbbq$ is doubly homogeneous, then the lattice $\mathcal{R}_f^\bbbq$ is trivial.
\end{proposition}
This is a generalization of  \cite{zheng2019computing} Theorem 3.2. The proof of this proposition is almost the same with the one given in \cite{zheng2019computing}, hence we omit it. Another generalization trough Galois-like groups of the ``Only If'' part of Proposition \ref{trivial} is given below:
\begin{proposition}\label{qtrigen}
Set $f\in\bbbq[x]$ $(f(0)\neq0)$ to be a polynomial without multiple roots and one of its roots is not a root of rational. Set $\mathcal{G}=\mathcal{G}_f$ $($respectively, $\mathcal{G}=\mathcal{G}_f^B$ or $\mathcal{G}_f^\bbbq)$ and $\mathcal{R}=\mathcal{R}_f\;(\text{resp., $\mathcal{R}=\mathcal{R}_f^\bbbq$})$. If $\mathcal{G}$ is transitive and the pair $(\mathcal{G},\mathcal{H})$, with $\mathcal{H}$ a root stabilizer, is $\bbbq$-trivial, then $\mathcal{R}$ is trivial.
\end{proposition}
\begin{proof}
Let $\mathcal{V}_1,\mathcal{V}_0$ and $\mathcal{W}=\bbbq\otimes \mathcal{R}$ be as in Proposition \ref{wfwfq}. For any $\sigma\in\mathcal{G}$, we define a coordinate permutation $\hat{\sigma}$ as in the proof of Proposition \ref{isgroup}. Then $\mathcal{W}$ is a $\bbbq[\mathcal{G}]$-submodule of $\bbbq^n$ by the definition of a Galois-like group (for any $v\in\bbbq^n$ or $v\in\mathcal{W}$, a group element $\sigma$ operates in the way so that it maps $v$ to the vector $\hat\sigma^{-1}(v)$). 

Since the pair $(\mathcal{G},\mathcal{H})$ is $\bbbq$-trivial, $\bbbq^n$ can be decomposed into two irreducible $\bbbq[\mathcal{G}]$-submodules: $\bbbq^n=\mathcal{V}_1\oplus\mathcal{V}_0$ (\cite{girstmair1999linear} Proposition 12). Since $\mathcal{G}$ is transitive, $\mathcal{V}_0$ and $\mathcal{V}_1$ are the only two irreducible $\bbbq[\mathcal{G}]$-submodules of $\bbbq^n$: 

Suppose that $\mathcal{V}\neq\mathcal{V}_0$ is an irreducible $\bbbq[\mathcal{G}]$-submodules and assume that $\mathcal{V}\cap\mathcal{V}_0\supsetneqq\{\mathbf{0}\}$. Then $\mathcal{V}\supsetneqq\mathcal{V}\cap\mathcal{V}_0$ or $\mathcal{V}_0\supsetneqq\mathcal{V}\cap\mathcal{V}_0$. This contradicts the fact that both $\mathcal{V}$ and $\mathcal{V}_0$ are irreducible, since $\mathcal{V}\cap\mathcal{V}_0\supsetneqq\{\mathbf{0}\}$ is a proper $\bbbq[\mathcal{G}]$-submodules of at least one of them. So we have $\mathcal{V}\cap\mathcal{V}_0=\{\mathbf{0}\}$. Noting that the $\bbbq$-dimension of $\mathcal{V}_0$ is $n-1$, one concludes that $\text{dim}_{\bbbq}(\mathcal{V})=1$. Set $v\in\mathcal{V}\backslash\{\mathbf{0}\}$, then $v\notin\mathcal{V}_0$ and $\sum_{i=1}^nv(i)\neq0$. Thus $\mathcal{V}\ni\sum_{\sigma\in\mathcal{G}}\hat\sigma^{-1}(v)=\frac{|\mathcal{G}|}{n}(\sum_{i=1}^nv(i))(1_{_1},\ldots,1_{_n})^T\neq\mathbf{0}$ follows from the transitivity of $\mathcal{G}$. Hence $\mathcal{V}=\mathcal{V}_1$.

Thus all the $\bbbq[\mathcal{G}]$-submodules of $\bbbq^n$ are $\{\mathbf{0}\}$, $\mathcal{V}_1$, $\mathcal{V}_0$ and $\bbbq^n$ itself. If $\mathcal{V}_0\subset\mathcal{W}$, then  $\mathcal{W}\subset\mathcal{W}_f^\bbbq$ and $\mathcal{V}_1\subset\mathcal{W}_f^\bbbq$ imply that $\mathcal{W}_f^\bbbq=\bbbq^n$, which contradicts the assumption that $f$ has a root that is not a root of rational. Hence $\mathcal{V}_0\not\subset\mathcal{W}$, which means $\mathcal{W}=\{\mathbf{0}\}$ or $\mathcal{W}=\mathcal{V}_1$. Thus $\mathcal{W}$ is trivial and so is the lattice $\mathcal{R}$.\qed
\end{proof}
\subsection{Necessary and Sufficient Condition for Exponent Lattice Triviality}
In this subsection, we characterize those polynomials $f$ with a trivial exponent lattice by giving a necessary and sufficient condition through the concept of a Galois-like group.
\begin{theorem}\label{rftri}
Set $f\in\bbbq[x]$ $(f(0)\neq0)$ to be a polynomial without multiple roots. Denote by $\beta_1,\ldots,\beta_t$ the rational roots of $f$ (if there are any) and by $\beta_0$ the rational number which is the product of all non-root-of-rational roots of $f$ (if there are any). Then the lattice $\mathcal{R}_f$ is trivial iff all the following conditions hold:

(i) the Galois-like group $\mathcal{G}_f=\Sigma$;

(ii) any root of $f$ is rational or non-root-of-rational;

(iii) the lattice $\mathcal{R}_{v_f}$ is trivial with the vector $v_f$ given by:
\[
v_f=\left\{ {\begin{array}{*{20}{ll}}
(\beta_0,\beta_1,\ldots,\beta_t)^T,&\text{ if $f$ has both rational and}\\
&\;\;\;\;\;\text{non-root-of-rational roots},\\
(\beta_1,\ldots,\beta_t)^T,&\text{ if any root of $f$ is rational},\\
(\beta_0),&\text{ if any root of $f$ is non-root-of-rational}.
\end{array}} \right.
\]
\end{theorem}
\begin{proof}
``If'': When $\deg(f)=1$, $\mathcal{R}_f$ is trivial and we are done. Suppose in the following that $\deg(f)\geq2$. Then the pair $(\mathcal{G}_f,\mathcal{H})=(\Sigma,\mathcal{H})$ is doubly homogeneous thus also $\bbbq$-trivial for any root stabilizer $\mathcal{H}$. If $f$ has a root that is not a root of rational, then $\mathcal{R}_f$ is trivial by Proposition \ref{qtrigen}. When all the roots of $f$ are rational, the lattice $\mathcal{R}_f=\mathcal{R}_{v_f}$ is trivial.

``Only If'': Now that $\mathcal{R}_f$ is trivial, the condition (i) is straightforward. Suppose that $f$ has a root $r$ which is a root of rational but not a rational number. Then the conjugations of $r$, say, $\{r=r^{(1)},r^{(2)},\ldots,r^{(s)}\}$, with $s\geq2$, are all the  roots of $f$. Then there is a positive integer $m$ so that $(r/r^{(2)})^m=1$. Thus $\mathcal{R}_f$ is non-trivial, which contradicts the assumption. So the condition (ii) holds. Since any nontrivial multiplicative relation of the vector $v_f$ results in a nontrivial multiplicative relation between the roots of $f$, the condition (iii) holds.\qed
\end{proof}

From the ``If'' part of the proof we observe that, when restricted to polynomials $f$ with degree higher than one, the condition (i) in Theorem \ref{rftri} can be replaced by the statement ``$\mathcal{G}_f$ is transitive and the pair $(\mathcal{G}_f,\mathcal{H})$ is $\bbbq$-trivial for any root stabilizer $\mathcal{H}$''. An interesting result follows directly from this observation:
\begin{corollary}\label{4eq}
Let $f$ be as in Theorem \ref{rftri}. If $\deg(f)\geq2$ and the conditions (ii)\,--\,(iii) in  Theorem \ref{rftri} hold, then the following conditions are equivalent to each other:

(i) $\mathcal{G}_f=\Sigma$;

(ii) $\mathcal{G}_f$ is doubly transitive;

(iii) $\mathcal{G}_f$ is doubly homogeneous;

(iv) $\mathcal{G}_f$ is transitive and the pair $(\mathcal{G}_f,\mathcal{H})$ is $\bbbq$-trivial for any root stabilizer $\mathcal{H}$.
\end{corollary}
\begin{proof}
The implications $(i)\Rightarrow(ii)$ and $(ii)\Rightarrow(iii)$ are trivial. The implication $(iv)\Rightarrow(i)$ follows from the ``If'' part of the proof of  \ref{rftri}. Now we prove the implication $(iii)\Rightarrow(iv)$: If $\deg(f)=2$ and $f$ has only rational roots, the lattice $\mathcal{R}_f=\mathcal{R}_{v_f}$ is trivial. So $\mathcal{G}_f=\Sigma$. If $\deg(f)=2$ but $f$ has a root that is not a root of rational, $f$ is irreducible over $\bbbq$ and $\Sigma=G_f\leq\mathcal{G}_f$ with $G_f$ the Galois group of $f$. In either case $\mathcal{G}_f$ is transitive. When $\deg(f)\geq 3$, the transitivity of $\mathcal{G}_f$ follows from \cite{dixon1996permutation} Theorem 9.4A. Now the $\bbbq$-triviality of the pair $(\mathcal{G},\mathcal{H})$ follows from Proposition \ref{2homoqtri}.\qed
\end{proof}
Thus the condition (i) of Theorem \ref{rftri} can be replaced by any one of the conditions (ii)--(iv) in Corollary \ref{4eq}.

For the lattice $\mathcal{R}_f^\bbbq$, we have a similar result:
\begin{theorem}\label{rfqtri}
Set $f\in\bbbq[x]$ $(f(0)\neq0)$ to be a polynomial without multiple roots. Then the lattice $\mathcal{R}_f^\bbbq$ is trivial iff all the following conditions hold:

(i) the Galois-like group $\mathcal{G}_f^B=\Sigma$;

(ii) either $\deg(f)=1$ or any root of $f$ is not a root of rational;

(iii) $f$ is irreducible over $\bbbq$.
\end{theorem}
\begin{proof}
``If'': When $\deg(f)=1$, this is trivial. Suppose in the following that $\deg(f)\geq2$ and any root of $f$ is not a root of rational. Then $\mathcal{G}_f^B=\Sigma$ is transitive and doubly homogeneous. Thus the pair $(\mathcal{G}_f^B,\mathcal{H})$ is $\bbbq$-trivial for any root stabilizer $\mathcal{H}$ by Proposition \ref{2homoqtri}. So $\mathcal{R}_f^\bbbq$ is trivial by Proposition  \ref{qtrigen}.

``Only If'': Now that $\mathcal{R}_f^\bbbq$ is trivial, it is clear that $\mathcal{G}_f^B=\Sigma$.  Suppose on the contrary that $f$ is reducible and $g_1$, $g_2$ are two of its factors. Let $\alpha_1,\ldots,\alpha_s$ denote the roots of $g_1$ and $\gamma_1,\ldots,\gamma_t$ the ones of $g_2$. Then for any two distinct integers $k$ and $l$, $(\alpha_1\ldots\alpha_s)^k(\gamma_1\ldots\gamma_t)^l\in\bbbq$. This contradicts the assumption that $\mathcal{R}_f^\bbbq$ is trivial. So $f$ is irreducible. Assume that $\deg(f)\geq2$ and one of the roots $r$ of $f$ is a root of rational. The conjugations of $r$, say, $\{r=r^{(1)},r^{(2)},\ldots,r^{(n)}\}$ ($n\geq2$) are exactly all the roots of $f$. Then there is a positive integer $m$ so that $(r/r^{(2)})^m=1$. Thus $\mathcal{R}_f$ is non-trivial and so is the lattice $\mathcal{R}_f^\bbbq$. This contradicts the assumption.\qed
\end{proof}

\begin{remark}
Theorem \ref{rfqtri} still holds when the equality $\mathcal{G}_f^B=\Sigma$ is replaced by $\mathcal{G}_f^\bbbq=\Sigma$ in the condition (i). The proof is almost the same. Moreover, from the ``If'' part of the proof we observe that, when restricted to polynomials $f$ with degree higher than one, the condition (i) in Theorem \ref{rfqtri} can be replaced by the statement ``$\mathcal{G}_f^B$ is transitive and the pair $(\mathcal{G}_f^B,\mathcal{H})$ is $\bbbq$-trivial for any root stabilizer $\mathcal{H}$'' or the statement ``$\mathcal{G}_f^\bbbq$ is transitive and the pair $(\mathcal{G}_f^\bbbq,\mathcal{H})$ is $\bbbq$-trivial for any root stabilizer $\mathcal{H}$''.
\end{remark}

The counterpart of Corollary \ref{4eq} in this case is given below:
\begin{corollary}\label{4eqagain}
Let $f$ be as in Theorem \ref{rfqtri} and $\mathcal{G}\in\{\mathcal{G}_f^B,\mathcal{G}_f^\bbbq\}$. If $\deg(f)\geq2$ and the conditions (ii)\,--\,(iii) in  Theorem \ref{rfqtri} hold, then the following conditions are equivalent to each other:

(i) $\mathcal{G}=\Sigma$;

(ii) $\mathcal{G}$ is doubly transitive;

(iii) $\mathcal{G}$ is doubly homogeneous;

(iv) $\mathcal{G}$ is transitive and the pair $(\mathcal{G},\mathcal{H})$ is $\bbbq$-trivial for any root stabilizer $\mathcal{H}$.
\end{corollary}
\begin{proof}
The proof is similar to the one of Corollary \ref{4eq}.\qed
\end{proof}
Theorem \ref{rftri} and \ref{rfqtri} characterize, for the first time, the the polynomial $f$ with a trivial exponent lattice $\mathcal{R}_f$ or $\mathcal{R}_f^\bbbq$ with the help of the concept of a Galois-like group. The conditions (ii)--(iii) in both theorems can be decided every efficiently (by \S\,5.1 of \cite{preprint} and \S\,2.2.1 of \cite{issac}). However, an efficient algorithm deciding whether a Galois-like group, of a given polynomial $f$, equals the symmetry group $\Sigma$ or not is not available at present.
\section{Conclusion}
We characterize the polynomials with trivial exponent lattices through the Galois and the Galois-like groups. Based on the algorithm {\tt IsQtrivial}, we extensively improve the main algorithm in \cite{zheng2019computing} proving triviality of the exponent lattice of a generic polynomial (when the polynomial degree is large). In addition, a sufficient and necessary condition is given with the help of the concept of a Galois-like group, which turns out to be essential in the study on  multiplicative relations between the roots of a polynomial. Further study on Galois-like groups seems to be interesting and promising.
\bibliographystyle{splncs04}
\bibliography{lncs444}
\end{document}